\documentclass[11pt,draftcls,onecolumn]{IEEEtran}
\usepackage{ifpdf}
\ifpdf
	\usepackage[pdftex]{graphicx}
  \usepackage[update]{epstopdf}
\else
	\usepackage{graphicx}
\fi
\usepackage{times}
\usepackage{epsf}
\usepackage{epsfig}
\usepackage{amsmath}
\usepackage{amsfonts}
\usepackage{amssymb}
\usepackage{amstext}
\usepackage{latexsym}
\usepackage{color}
\usepackage{ifthen}
\usepackage{multirow}
\usepackage{subfigure}
\usepackage{cite,setspace}
\usepackage{algorithmic,algorithm}

\newcommand{\be}{\begin{equation}}
\newcommand{\ee}{\end{equation}}
\newcommand{\bear}{\begin{eqnarray}}
\newcommand{\eear}{\end{eqnarray}}
\newcommand{\bears}{\begin{eqnarray*}}
\newcommand{\eears}{\end{eqnarray*}}
\newcommand{\bi}{\begin{itemize}}
\newcommand{\ei}{\end{itemize}}
\newcommand{\ben}{\begin{enumerate}}
\newcommand{\een}{\end{enumerate}}
\newtheorem{theorem}{\textbf{Theorem}}
\newtheorem{prop}{\textbf{Proposition}}
 
\newtheorem{definition}{\textbf{Definition}}

\newtheorem{corollary}{\textbf{Corollary}}

\renewcommand{\vec}[1]{\mbox{\boldmath$#1$}}

\begin{document}

\title{Exact-Repair Regenerating Codes Via Layered Erasure Correction and Block Designs}

\author{Chao Tian,~\IEEEmembership{Senior Member,~IEEE}, Vaneet Aggarwal,~\IEEEmembership{Member,~IEEE}, \\and Vinay A. Vaishampayan,~\IEEEmembership{Fellow,~IEEE}
}
\maketitle




\begin{abstract}

A new class of exact-repair regenerating codes is constructed by combining two layers of erasure correction codes together with combinatorial block designs, {\em e.g.}, Steiner systems, balanced incomplete block designs and $t$-designs. The proposed codes have the \lq\lq{}uncoded repair\rq\rq{} property where the nodes participating in the repair simply transfer part of the stored data directly, without performing any computation. The layered error correction structure makes the decoding process rather straightforward, and in general the complexity is low. We show that this construction is able to achieve performance better than time-sharing between the minimum storage regenerating codes and the minimum repair-bandwidth regenerating codes. 
\end{abstract}

\section{Introduction}

Distributed data storage systems can encode and disperse information (a message) to multiple storage nodes (or disks) in the network such that a user can retrieve it by accessing only a subset of them. This kind of systems is able to provide superior reliability performance in the event of disk corruption or network congestion. In order to reduce the amount of storage redundancy required to guarantee such reliability performance, erasure correction codes can be used instead of simple replication of the data.

When the data is coded by an erasure code, data repair ({\em e.g.}, due to node failure) becomes more involved, because the information stored at a given node may not be directly available from any one the remaining storage nodes, but it can be nevertheless reconstructed since it is a function of the information stored at these nodes. One key issue that affects the system performance is the total amount of information that the remaining nodes need to transmit to the new node. Consider a storage system which has a total of $n$ storage nodes, and the data can be reconstructed by accessing any $k$ of them. A failed node is repaired by requesting any $d$ of the remaining nodes to provide information, and then using the received information to construct a new data storage node. A naive approach is to let these helper nodes transmit sufficient data such that the underlying complete data can be reconstructed, and then the information that needs to be stored at the new node can be subsequently generated. This approach is however rather wasteful, since the data stored at the new node is only a fraction of the complete data. 

Dimakis {\em et al.} in \cite{Dimakis:10} provided a theoretical framework, namely regenerating codes, to investigate the tradeoff between the amount of storage at each node ({\em i.e.,} data storage) and the amount of data transfer for repair ({\em i.e.,} repair bandwidth). It was shown that for the case when the regenerated information at the new node only needs to fulfill the role of the failed node functionally ({\em i.e.}, functional-repair), but not to replicate exactly the original information content at the failed node ({\em i.e.}, exact-repair), the problem can be converted to an equivalent network multicast problem, and thus the celebrated network coding result \cite{Yeung:00} can be applied.  By way of this equivalence, the optimal tradeoff between the storage and repair bandwidth was completely characterized in \cite{Dimakis:10} for functional-repair regenerating codes. The two important extreme cases of the optimal tradeoff, where the data storage is minimized and the repair bandwidth is minimized, are referred to minimum storage regenerating (MSR) codes and minimum bandwidth regenerating (MBR) codes, respectively.  The problem of functional-repair regenerating codes is well understood and constructions of such codes are available \cite{Dimakis:10,Wu:10,Dimakis:11}.

The functional-repair framework implies that the repair rule and the decoding rule in the system may evolve over time, which incurs additional system overhead. Furthermore, functional repair does not guarantee the data to be stored in systematic form, which is an important practical requirement to consider. In contrast, exact-repair regenerating codes do not have such disadvantages. The problem of exact-repair regenerating codes was investigated in \cite{RashmiShah:11,RashmiShah:12:1,RashmiShah:12:2,Cadambe:11,PapailiopoulosDimakisCadambe:11,Tamo:11,CadambeHuang:11}, all of which address either the MBR case or the MSR case. Particularly, the optimal code constructions in \cite{RashmiShah:12:1} and \cite{RashmiShah:11} show that the more stringent exact-repair requirement does not incur any penalty for the MBR case; the constructions in \cite{RashmiShah:12:2,RashmiShah:11,Cadambe:11} show that this is also true for the MSR case. These results may lead to the impression that enforcing exact-repair never incurs any penalty compared to functional repair. However, the result in \cite{RashmiShah:12:1} shows that this is not the case, and in fact a large portion of the optimal tradeoffs achievable by functional-repair codes can not be strictly achieved by exact-repair codes\footnote{One may question whether exact-repair codes can asymptotically approach these tradeoffs, however in \cite{Tian:12} it is shown that there indeed exists a non-vanishing gap between the optimal functional-repair tradeoff and the exact-repair tradeoff.}. 

Codes achieving tradeoff other than the MBR or the MSR points may be more suitable for a system employing exact-repair regenerating codes. From a practical point of view, codes achieving other tradeoff points may have lower complexity than using the time-sharing approach, because the MSR point requires interference alignment, and it is known to be impossible for linear codes to achieve the MSR point for some parameters without symbol extension \cite{RashmiShah:12:2}. As such, it is important to find such codes with competitive performance and low complexity. 
However, it is in fact unknown whether there even exist codes that can achieve a storage-bandwidth tradeoff better than simply time-sharing between an MBR code and an MSR code. 
In this work, we provide a linear code construction based on the combination of two layers of erasure correction codes and combinatorial block designs, which is indeed able to achieve tradeoff points better than the time-sharing between an MBR code and an MSR code. The two erasure correction codes are not independent, which must be jointly designed to satisfy certain full rank conditions to guarantee successful decoding. In this work we mainly focus on the case when $d=n-1$, {\em i.e.}, when the repair requires the access to all the other storage nodes, however it can indeed to generalized to the case $d<n-1$.

The conceptually straightforward code construction we propose has the property that the nodes participating in the repair do not need to perform any computation, but can simply transmit certain stored information for the new node to synthesis and recover the lost information. The uncoded repair property is appealing in practice, since it reduces and almost completely eliminates the computation burden at the helper nodes. This property also holds in the constructions proposed in \cite{RashmiShah:12:1} and \cite{Papailiopoulos:12INFOCOM}. In fact our construction was partially inspired by and may be viewed as a generalization of these codes. 
Another closely related work is \cite{Rouayheb:10}, where repetition and erasure correction codes are combined to construct codes for the MBR point, and one of constructions indeed relies on Steiner systems. The model in \cite{Rouayheb:10} is however different from ours (and that in \cite{Dimakis:10}), where the repair procedure only needs to guarantee the existence of one particular $d$-helper-node combination (fix-access repair), instead of the more stringent requirement that the repair information can come from any $d$-helper-node combination (random-access repair). Though both models have their merits, we focus on the more stringent and thus more robust random-access repair model in this work. 

The rest of the paper is organized as follows. In Section \ref{sec:definition}, a formal definition is given for the coding problem and several relevant existing results are reviewed. Section \ref{sec:example}  provides an example to illustrate the structure of the proposed construction. Section \ref{sec:dequalsnminus1} provides the general code construction in three progressive steps, and in Section \ref{sec:performance} the performance is analyzed. Finally \ref{sec:conclusion} concludes the paper.

\section{Problem Definition and Preliminaries}
\label{sec:definition}

In this section, we first provide a formal definition of exact-repair regenerating codes. 
Some existing results on regenerating codes, basics on maximum separable regenerating codes and block designs are also briefly reviewed. 

\subsection{Definition of Exact-Repair Regenerating Codes}

An $(n,k,d)$ exact-repair regenerating code is a storage system with a total of $n$ storage nodes (disks)\footnote{From here on, we shall use \lq\lq{}node\rq\rq{} and \lq\lq{}disk\rq\rq{} interchangably.}, where any $k$ of them can be used to reconstruct the complete data, and furthermore to repair a lost disk, the new disk may access data from any $d$ of the remaining $n-1$ disks. Let the total amount of raw data stored be $M$ units and let each storage site stores $\alpha$ units of data, which implies that the redundancy of the system is $n\alpha-M$. To repair a disk failure (regenerate a new disk), each contributing disk transmits $\beta$ units of data to the new node, which results in a total of $d\beta$ units of data transfer for repair. It is clear that the quantities $\alpha$ and $\beta$ scale linearly with $B$, because a code can simply be concatenated. For this reason we shall normalize them the other two quantities using $\beta$
\begin{align}
\bar{\alpha}\triangleq \frac{\alpha}{\beta},\quad \bar{M}\triangleq\frac{M}{\beta}, 
\end{align}
and use them as the measure of performance from here on. 

Formally, the problem can be defined as follows. The notation $I_n$ is used to denote the set $\{1,2,\ldots,n\}$, and without loss of generality we assume $k\leq d$.
\begin{definition}
An $(n,k,d,N,N_d,K)$ exact-repair regenerating code consists of a total of $n$ encoding function $f^E_i(\cdot)$, a total of ${n \choose k}$ decoding functions $f^D_{A}(\cdot)$, a total of $nd{n-1 \choose d}$ repair encoding functions $F^{E}_{i,A,j}(\cdot)$,  and a total of $n{n-1 \choose d}$ repair decoding functions $F^{D}_{i,A}(\cdot)$, where
\begin{align}
f^E_i:I_N\rightarrow I_{N_d},\quad i\in I_n,
\end{align}
which map the message $m\in I_N$ to $n$ pieces of coded information, 
\begin{align}
f^D_{A}:I^k_{N_d}\rightarrow I_N,\quad A\subset {I}_n\quad \mbox{and}\quad |A|=k
\end{align}
which maps the $k$ pieces of coded information in a set $A$ to the original message, 
\begin{align}
F^{E}_{i,A,j}:I_{N_d}\rightarrow I_{K},\quad j\in I_n,\quad A\subseteq {I}_n\setminus\{j\}\quad \mbox{and}\quad |{A}|=d,\quad \, i\in {A},
\end{align}
which maps a piece of coded information to an index that will be made available to the new node, and
\begin{align}
F^{D}_{j,{A}}:{I}^d_{K} \rightarrow {I}_{N_d},\quad \quad j\in{I}_n,\quad{A}\subseteq {I}_n\setminus\{j\}\quad \mbox{and}\quad |{A}|=d,
\end{align}
which maps $d$ of such indices from the helper nodes to reconstruct the information stored at the lost node.  The functions must satisfy the data reconstruction conditions
\begin{align}
f_{{A}}^D\left(\prod_{i\in{A}}f^E_i(m)\right)=m,\quad m\in{I}_N,\quad{A}\subset {I}_n\quad \mbox{and}\quad |{A}|=k,
\end{align}
and the repair conditions
\begin{align}
F^D_{j,{A}}\left(\prod_{i\in{A}}F^E_{i,{A},j}\left(f^E_i(m)\right)\right)=f^E_j(m),\quad m\in{I}_N,\quad j\in{I}_n,\quad {A}\subseteq{I}_n\setminus\{j\}\quad \mbox{and}\quad |{A}|=d.
\end{align}
\end{definition}

\begin{definition}
A normalized pair $(\bar{\alpha},\bar{M})$ is said to be achievable for $(n,k,d)$ regenerating if for any $\epsilon>0$ there exists an $(n,k,d,N,N_d,K)$ code such that
\begin{align}
\bar{\alpha}+\epsilon\geq \frac{\log N_d}{\log K}
\end{align}
and
\begin{align}
\bar{M}-\epsilon\leq \frac{\log N}{\log K}.
\end{align}
\end{definition}

The quantity $\epsilon$ in the definition above is introduced to include the case when the storage-bandwidth tradeoff may be approached asymptotically, {\em e.g.}, the case discussed in \cite{Cadambe:11}. 

It is sometimes insightful to consider the case when $n$ is large while $k=n-\tau_1$ and $d=n-\tau_2$ where $\tau_1$ and $\tau_2$ are fixed positive constant integers such that $\tau_1\geq \tau_2$.  For this purpose, the following two quantities become relevant.
\begin{definition}
An $E$-pair $(\mathsf{E}^{(n)}_r,\mathsf{E}^{(n)}_d)$ where
\begin{align} 
\mathsf{E}^{(n)}_r\triangleq\frac{\log(\bar{M}-n\bar{\alpha})}{\log n},\quad \mathsf{E}^{(n)}_d\triangleq\frac{\log\bar{M}}{\log n}
\end{align}
is $(n,\tau_1,\tau_2)$-achievable if $(\bar{\alpha}^{(n)},\bar{M}^{(n)})$ is achievable for $(n,n-\tau_1,n-\tau_2)$ regenerating. The collection of all $(n,\tau_1,\tau_2)$-achievable pairs is denoted as $\mathcal{E}^{(n)}$. The achievable redundancy-data-rate exponent region $\mathcal{E}$ is  the closure of  $\limsup_{n\rightarrow \infty}\mathcal{E}^{(n)}$.
%
\end{definition}

In Section \ref{sec:performance}, we shall show that the proposed codes are able to achieve the entire exponent region $\mathcal{E}$, while time-sharing between the MSR point and the MBR point can not.

\subsection{Cut-Set Outer Bound, MBR Point and MSR Point}
\begin{figure}[tcb]
  \centering
  \includegraphics[width=10cm]{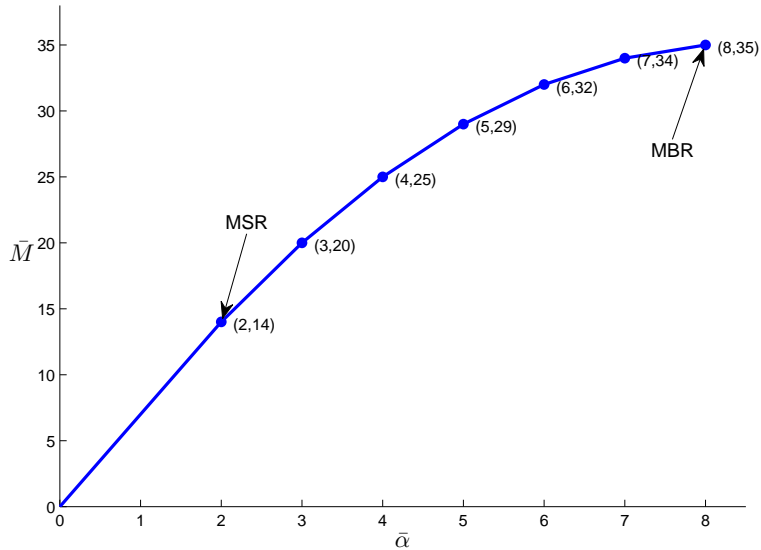}
  \caption{Cutset bound for $(n,k,d)=(9,7,8)$.\label{fig:cutset}}
\end{figure}

As mentioned earlier, the functional-repair regenerating coding problem can be converted to a multicast problem, and through this connection, a precise characterization of the optimal storage-bandwidth tradeoff was obtained in \cite{Dimakis:10} using cut-set analysis. Since exact-repair is a more stringent requirement than functional-repair, this characterization provides an outer bound for exact-repair regenerating codes.
\begin{theorem}[\cite{Dimakis:10}]
\label{theorem:cutset}
Any exact-repair regenerating codes must satisfy the following condition
\begin{align}
\sum_{i=0}^{k-1}\min(\bar{\alpha},(d-i))\geq \bar{M}.
\end{align}
\end{theorem}

One extreme case of this outer bound is when the storage is minimized, {\em i.e.}, the minimum storage regenerating (MSR) point, which is
\begin{align}
\bar{\alpha}=(d-k+1), \quad \bar{M}=k(d-k+1).
\end{align}
The other extreme case is when the repair bandwidth is minimized, {\em i.e.}, the minimum bandwidth regenerating (MBR) point, which is
\begin{align}
\bar{\alpha} =d,\quad\bar{M}=\frac{k(2d-k+1)}{2}.
\end{align}
Both of these two extreme points are achievable \cite{RashmiShah:11,RashmiShah:12:1,RashmiShah:12:2,Cadambe:11} also for the exact-repair case. The functional repair outer bound is however not tight in general, which implies that the exact-repair condition will indeed incur a penalty in many cases \cite{RashmiShah:12:1}\cite{Tian:12}. The cut-set outer bound and the two extreme points are illustrated in Fig. \ref{fig:cutset} for $(n,k,d)=(9,7,8)$; note that the bound is piece-wise linear. The segment between the MSR point and the origin $(0,0)$ is given by the trivial bound $k\bar{\alpha}\leq \bar{M}$, and it is essentially a degenerate regime because to achieve this segment of tradeoff, we can simply utilize an MSR code but let the helper nodes send more than necessary amount ({\em i.e.,} more than $\beta$ units) of data. 


\subsection{Maximum Distance Separable Code}

A linear code of length-$n$ and dimension $k$ is called an $(n,k)$ code. The Singleton bound (see {\em e.g.}, \cite{Wicker:book}) is a well known upper bound on the minimum distance for any $(n,k)$ code.

\begin{theorem}
The minimum distance $d_{\min}$ for an $(n,k)$ code is bounded by $d_{\min}\leq n-k+1$.
\end{theorem}

An $(n,k)$ code that satisfies the Singleton bound with equality is called a maximum distance separable (MDS) code. A key property of an MDS code is that it can correct any $(n-k)$ or less erasures. There are many ways to find MDS codes for any given $(n,k)$ values, $n\geq k$. For example, any randomly generated $n\times k$ matrix in a sufficiently large alphabet is a generator matrix for an MDS code with high probability. Any $n\times k$ Vandermonde matrix can also generate an MDS code when the entries in the second column are all distinct. Another explicit construction approach is by puncturing a Reed-Solomon code of an appropriate alphabet (see {\em e.g.}, \cite{Wicker:book}).

\subsection{Block Designs}

Block design has been considered in combinatorial mathematics with applications in experimental design, finite geometry, software testing, cryptography, and algebraic geometry. Generally speaking, a block design is a set together with a family of subsets ({\em i.e.}, blocks) whose members are chosen to satisfy some properties that are deemed useful for a particular application. Usually the blocks are required to all have the same number of elements, and in this case a given block design with parameter $(n,k)$ is specified by $(X,\mathcal{B})$ where $X$ is an $n$-element set and $\mathcal{B}$ is a collection of $k$-element subsets of $X$. 

One important class of block designs is the $t$-designs. The class of $t$-designs with parameter $(\lambda,t,r,n)$ is denoted as $S_{\lambda}(t,r,n)$; a valid $t$-design in $S_{\lambda}(t,r,n)$ is a pair $(X,\mathcal{B})$ where $X$ is $n$-element set and $\mathcal{B}$ is a collection of $r$-element subsets of $X$ with the property that every element in $X$ appears in exactly $\gamma$ blocks and every $t$-element subset of $X$ is contained in exactly $\lambda$ blocks. Without loss of generality, one can always use $X=I_n$, and we shall use this convention from here on.

The most extensively researched class of block designs is perhaps Steiner systems, which is the case when $\lambda=1$ and $t\geq 2$. In this case, the subscript $\lambda$ is usually omitted and we directly write it as $S(t,r,n)$. The simplest design in this class is when $t=2$ and $r=3$, which is the particularly well understood Steiner triple systems $S(2,3,n)$. It is known that there exists a Steiner triple system $S(2,3,n)$ if and only if $n=0$, or $n$ modulo $6$ is $1$ or $3$; see, {\em e.g.}, \cite{Colbourn:book}. It follows that the smallest positive integer which gives us a non-trivial Steiner system is $n=7$ and the next is $n=9$. Examples of $S(2,3,7)$, $S(2,3,9)$ are given in Table \ref{table:steiner79}, where a design for $S(2,4,13)$ is also included. Another well-known special class of $t$-designs is Balanced Incomplete Block Designs (BIBDs), which is a special case of $t$-designs for the case $t=2$. It is clear that Steiner systems $S(2,3,n)$ are also BIBDs. 
\begin{table}[t]
\caption{Example Steiner triple systems $S(2,3,7)$, $S(2,3,9)$ and $S(2,4,13)$.}
\label{table:steiner79}
\centering
\begin{tabular}{|c|c|}
\hline
$({I}_7,\mathcal{B})\in S(2,3,7)$&$\{(1,2,3),(1,4,5),(1,6,7),(2,4,6),(2,5,7),(3,4,7),(3,5,6)\}$\\\hline
$({I}_9,\mathcal{B})\in S(2,3,9)$&$\{(2,3,4),(5,6,7),(1,8,9),(1,4,7),(1,3,5),(4,6,8),$\\
&$(2,7,9),(2,5,8),(1,2,6),(4,5,9),(3,7,8),(3,6,9)\}$\\\hline
$({I}_{13},\mathcal{B})\in S(2,4,13)$&$\{(1,2,4,10),(2,3,5,11),(3,4,6,12),(4,5,7,13),(5,6,8,1),$\\
&$(6,7,9,2),(7,8,10,3),(8,9,11,4),(9,10,12,5),$\\
&$(10,11,13,6),(11,12,1,7),(12,13,2,8),(13,1,3,9)\}$\\\hline
\end{tabular}
\end{table}


For a given $(\lambda,t,r)$ triple, a $t$-design may not exist for an arbitrary $n$, however, for any $(t,r,n)$, a trivial $t$-design always exists with $\lambda^*(t,r,n)\triangleq{{n-t} \choose {r-t}}$, as given in the following proposition.
\begin{prop}
For any $(t,r,n)$ where $t\leq r\leq n$, a complete block design is a block design where the blocks are all the $r$-element subsets of ${I}_n$ (and the blocks are not repeated). In this design every element in ${I}_n$ appears in exactly ${{n-1}\choose {r-1}}$ blocks and every $t$-element subset of $X$ is contained in exactly $\lambda^*(t,r,n)$ blocks.
\end{prop}

We may still refer to such a complete block design for the case of $t=2$ as a BIBD, although it is in fact a complete block design instead of an incomplete one. The following proposition \cite{Colbourn:book} is useful.



\begin{prop}
\label{theorem:numberofblocks}
If $({I}_n,\mathcal{B})$ is an $S_\lambda(t,r,n)$ design and $S$ is any $s$-element subset of ${I}_n$, with $0\leq s\leq t$, then the number of blocks containing $S$ is
\begin{align}
|\{B\in \mathcal{B}:S\subseteq B\}|=\lambda{{n-s \choose t-s}}{{r-s \choose t-s}}^{-1}.
\end{align}
\end{prop}

The following corollary apparently follows by setting $s=0$ in Theorem \ref{theorem:numberofblocks}.
\begin{corollary}
If $(I_n,\mathcal{B})$ is a $S_\lambda(t,r,n)$ design, then the total number of blocks in $\mathcal{B}$ is
\begin{align}
N_\lambda(t,r,n)\triangleq|\mathcal{B}|=\lambda{n \choose t}{r \choose t}^{-1}.
\label{eqn:N}
\end{align}
\end{corollary}

For the case of Steiner systems, we shall omit the subscript, and simply write it as $N(t,k,n)$. When the parameters are clear from the context, we may also write $N_\lambda(t,r,n)$ as $N^*$.

There are various known constructions,  existence results, and non-existence results for Steiner systems, BIBDs and $t$-designs in the literature; interested readers are referred to \cite{Colbourn:book} and  \cite{Bose:39} for more details.

\section{An Example $(9,7,8)$ Code}
\label{sec:example}

To illustrate the basic code components, we shall construct a $(9,7,8)$ exact-repair regenerating code with $M=23$, $\alpha=4$ and $\beta=1$.  The addition and multiplication operations in the encoding and decoding are in the finite field $\mathbb{F}(3)$, however this choice is only for better concreteness. The construction is based on the block design $S(2,3,9)$ given in Table \ref{table:steiner79}.

Let the information be given as a length-$23$ vector, where the $i$-th entry is denoted as $d_i\in \mathbb{F}(q)$. The components of this code are given described below. 

\textbf{Encoding:}

\begin{enumerate}
\item Generate a parity symbol $d_{24}=\sum_{j=1}^{12}d_{2j-1}+\sum_{j=1}^{11}2d_{2j}$;
\item Pair up $(d_{2j-1},d_{2j})$, and rename it as $(X_j,Y_j)$ where $j=1,2,\ldots,12$; {\em i.e.}, $(X_j,Y_j)\triangleq (d_{2j-1},d_{2j})$.
\item Generate a new parity symbol $P_{j}=X_{j}+Y_{j}$, and $(X_j,Y_j,P_j)$ will be referred to as a parity group, where $j=1,2,\ldots,12$;
\item For each block $B_j=\{b_{j,1},b_{j,2},b_{j,3}\}$, $j=1,2,\ldots,12$, in the block design, write one symbol in $j$-th parity group in the $b_{j,1}$-th disk, one in the $b_{j,2}$-th disk, and one in the $b_{j,3}$-th disk, respectively.
\end{enumerate}

One possible resulting code symbol placement is illustrated in Table. \ref{table:example9}. The placement is not unique, since within each parity group, the symbols can be permuted arbitrarily. Note that the second step above is for facilitating better understanding, and a more concise set of notations will be used in the general construction in the next section. 

\begin{table}[t]
\caption{Code constructed using the Steiner triple system $S(2,3,9)$ in Table \ref{table:steiner79}}
\label{table:example9}
\centering
\begin{tabular}{|c|c|c|c|c|c|c|c|c|c|}
\hline
Disk \#&1&2&3&4&5&6&7&8&9\\
\hline
&$P_3$&$X_1$&$Y_1$&$P_1$&$X_2$&$Y_2$&$P_2$&$X_3$&$Y_3$\\
&$X_4$&$X_7$&$Y_5$&$Y_4$&$P_5$&$Y_6$&$P_4$&$P_6$&$Y_7$\\
&$X_5$&$X_8$&$X_{11}$&$X_6$&$Y_8$&$P_9$&$P_7$&$P_{8}$&$P_{10}$\\
&$X_9$&$Y_9$&$X_{12}$&$X_{10}$&$Y_{10}$&$Y_{12}$&$Y_{11}$&$P_{11}$&$P_{12}$\\
\hline
\end{tabular}
\end{table}

\textbf{Repair:} 

Let us suppose the first disk fails. To regenerate, for example, symbol $X_5$, first obtain $Y_5$ and $P_5$ from disk-3 and disk-5, respectively, and then compute $X_5=P_5-Y_5$. Clearly, other symbols on the disk can also be repaired following a similar procedure. This procedure also applies to other disk failures. It can also be checked that for any disk failure, each remaining disk sends a single symbol during the repair, which is in fact guaranteed by the basic property of  block designs in this case. 


\textbf{Reconstruction:} 

For data reconstruction, several different cases need to be considered. Before going into the details of these cases, consider a scenario where disk $1$ and disk $2$ are not accessible. Notice that although $(X_1,X_4,X_5,X_7,X_8,X_9,Y_9,P_3)$ are not accessible directly, $(X_1,X_4,X_5,X_7,X_8,P_3)$ can be recovered using the symbols on other disks, as discussed in the repair procedure above; thus parity groups $1,2,3,4,5,6,7,8$ are not effected. As a consequence, only symbols in the $9$-th parity group can not be completely recovered, but even in this parity group, $P_9$ is still accessible on disk-$6$. The reconstruction cases can be classified according to which parity group is effected ({\em i.e.}, cannot be completely recovered directly) and which symbol within this parity group is still accessible. 
\begin{enumerate} 
\item The $j$-th parity group, $i\in I_{11}$, is effected, but $X_j$ or $Y_j$ is still accessible within it. An example case is when disk-$3$ and disk-$4$ are not accessible. Note $X_1=d_1$ is still available in this case, and $d_1+2 d_{2}$ can be computed using $d_{24}=Y_{12}=\sum_{j=1}^{12}d_{2j-1}+\sum_{j=1}^{11}2d_{2j}$ after eliminating $(d_{2j-1},d_{2j})$ pairs for $j=2,3,\ldots,11$ and $d_{23}$, from which $(d_1,d_{2})$ can be solved.  The information vector can be obtained by rearranging the symbols.
\item The $j$-th parity group, $j\in I_{11}$, is effected, but the parity symbol $P_j$ is still accessible within it. An example case is when disk-$2$ and disk-$3$ are not accessible. In this case $(d_{2j-1},d_{2j})$ pairs for $j=2,3,\ldots,12$ can be recovered. In addition, $P_1=d_1+d_{2}$ and $d_1+2 d_{2}$ are available, from which $(d_1,d_{2})$ can be solved.
\item Parity group $12$ is effected, but $X_{12}$ is still accessible. This case is trivial since all $d_j$, $j=1,2,\ldots,23$ have been directly recovered.
\item Parity group $12$ is effected, but $Y_{12}$ is still accessible. In this case $(X_j,Y_j)=(d_{2j-1},d_{2j})$ for $j=1,2,\ldots,11$ can be recovered, and thus only $d_{23}$ needs to be recovered. But we have $Y_{12}=d_{24}=\sum_{i=1}^{12}d_{2j-1}+\sum_{1}^{11}2d_{2j}$, from which $d_{23}$ can now be obtained. 
\item Parity group $12$ is effected, but $P_{12}$ is still accessible. Again $(d_{2j-1},d_{2j})$ pairs for $j=1,2,\ldots,11$ can be recovered. Additionally we have $P_{12}=d_{24}+d_{23}=\sum_{j=1}^{11}d_{2j-1}+\sum_{j=1}^{11}2d_{2j}+2d_{23}$, from which $d_{23}$ can be obtained. 
\end{enumerate}

Let us compare this code with the time-sharing code using an MBR code and an MSR code. For $(n,k,d)=(9,7,8)$, the MSR point is $(\bar{\alpha},\bar{M})=(2,14)$ and the MBR point is 
$(\bar{\alpha},\bar{M})=(8,35)$. Our construction achieves $(\bar{\alpha},\bar{M})=({4},{23})$, while the time sharing performance between the MBR point and the MSR point at $\bar{\alpha}=4$ gives $\bar{M}=21$, thus the example construction indeed achieves an improvement on $\bar{M}$ while keeping $\bar{\alpha}$ the same.

This example illustrates the main components in the proposed construction, {\em i.e.}, a block design, a first layer long MDS code, and a second layer short MDS code. The coefficients used in the two parity symbols of the two codes cannot be set arbitrarily, for example, if we were to set $P_{j}=X_{j}+2Y_{j}$, then in the second case  discussed in the reconstruction procedure, a decoding failure would occur. The basic idea is to use the short MDS code to recover as many data symbols as possible which will render most of the parity symbols in the short MDS code redundant, and then use the remaining parity symbol in the short MDS code together with the parity symbol in the long MDS code to jointly solve the remaining unknown data symbol. 

\section{Code Constructions}
\label{sec:dequalsnminus1}

In this section, we first describe an explicit code construction for $(n,n-2,n-1)$ code based on Steiner system $S(2,r,n)$. This construction however only applies to the case when a Steiner system exists for such $n$, and as aforementioned, Steiner systems may not exist for all $(r,n)$ pairs. Then based on BIBDs $S_\lambda(2,r,n)$, the method is generalized to the case any $(n,k,d)$ triples such that $k\leq n-1$ and $d=n-1$. Since a complete block design can be viewed as a special case of BIBDs, the construction applies to any value of positive integer $n$. This construction can be further generalized to the case when $d<n-1$, which will be discussed briefly. 

\subsection{A Construction Based on $S(2,r,n)$}

Given a block design $({I}_n,\mathcal{B})\in S(2,r,n)$, the exact-repair regenerating code with parameters $(n,k,d)=(n,n-2,n-1)$ we shall construct has the following parameters 
\begin{align}
\alpha=\frac{n-1}{r-1},\quad\beta=1,\quad M=(r-1)N^*-1=\frac{n(n-1)}{r}-1,
\end{align}
where we have used $N^*$ to denote $N(2,r,n)$ for notational simplicity.
Note that these parameters are all integers for a valid Steiner system, moreover, $n(n-1)$ is a multiple of $r(r-1)$, which can be seen using Proposition \ref{theorem:numberofblocks} and its corollary. The alphabet for this code can be chosen to be any finite field $\mathbb{F}(q)$ with a field size $q\geq r$, and the addition and multiplication operations in the encoding and decoding process are performed in this field. 

Let the $M$ information symbols in $\mathbb{F}(q)$ be given in a $(r-1)\times N^*$ matrix except the bottom-right entry $D_{r-1,N^*}$, which is left blank.
The code has several components:

\textbf{Encoding:}

\begin{enumerate}
\item Choose $(r-1)$ distinct non-zero elements $\phi_1,\phi_2,\ldots,\phi_{r-1}$ in $\mathbb{F}(q)$, which satisfy $\phi_i+1\neq 0$ for $i=1,2,\ldots,r-2$. Generate a parity symbol and assign it to $D_{r-1,N^*}$ as
\begin{align}
\label{eqn:longMDS}
D_{r-1,N^*}=\sum_{i=1}^{r-2}\phi_i\sum_{j=1}^{N^*} D_{i,j}+\phi_{r-1}\sum_{j=1}^{N^*-1}D_{r-1,j}.
\end{align}
\item For each column $j=1,2,\ldots,N^*$, generate new parity symbols as 
\begin{align}
D_{r,j}\triangleq P_{j}=\sum_{i=1}^{r-1}D_{i,j}.
\label{eqn:paritygroup}
\end{align} 
The collection $(D_{1,j},D_{2,j},\ldots,D_{r-1,j},P_j)$ will be referred to as the $j$-th parity group;
\item For each block $B_j=\{b_{j,1},b_{j,2},\ldots,b_{j,r}\}\in \mathcal{B}$, $j=1,2,\ldots,N^*$, distribute the symbols in the $i$-th parity group onto disk $b_{j,1},b_{j,2},\ldots,b_{j,r}$, one symbol onto each disk. 
\end{enumerate}

\begin{figure}[tcb]
  \centering
  \includegraphics[width=10cm]{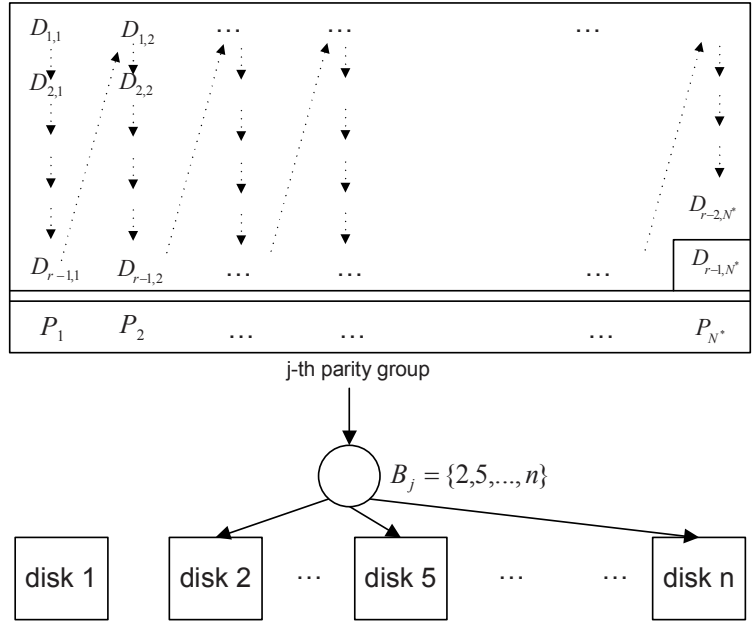}
  \caption{Code structure based on  $S(2,r,n)$.\label{fig:code1}}
\end{figure}

\textbf{Repair:} 

Suppose disk-$m$ fails. In order to recover the symbols on this disk, find in $\mathcal{B}$ all blocks $B_j$ such that $m\in B_{j}$. Recall there are a total of $\alpha$ such blocks, and let them be denoted as $B_{k_1},B_{k_2},\ldots,B_{k_\alpha}$. For each of this block $B_{k_l}$, $l=1,2,\ldots,\alpha$, obtain the symbols in the parity group $k_l$ from the disks in the set $B_{k_l}\setminus \{m\}$, and recover the symbol in this parity group on disk-$j$ using the relation (\ref{eqn:paritygroup}). 

\textbf{Reconstruction:} 

Several cases need to be considered, when two disks have failed:
\begin{enumerate} 
\item The $j$-th parity group loses two symbols which are the parity symbol $P_j$ and one data symbol $D_{i,j}$, and the other parity groups each lose one symbol or less. This implies that the other parity groups can recover all its data using (\ref{eqn:paritygroup}), and thus only $D_{i,j}$ needs to be recovered. It can be obtained through $D_{r-1,N^*}$, by eliminating  in (\ref{eqn:longMDS}) the symbols in the other parity group, and then eliminating $D_{k,j}$, $k\neq i$. 

\item The $j$-th parity group loses two symbols which are two data symbols $D_{i_1,j}$ and $D_{i_2,j}$, and the other parity groups each lose one symbol or less. 
Other data symbols can be obtained as in the previous case, and only $D_{i_1,j}$ and $D_{i_2,j}$ need to be recovered. Since $D_{r-1,N^*}$ is still available, by eliminating the symbols in the other parity group in (\ref{eqn:longMDS}), and then eliminating $D_{k,j}$, $k\neq i_1$ and $k\neq i_2$, we obtain $\phi_{i_1}D_{i_1,j}+\phi_{i_2}D_{i_2,j}$. By eliminating $D_{k,j}$, $k\neq i_1$ and $k\neq i_2$ in (\ref{eqn:paritygroup}), we obtain $D_{i_1,j}+D_{i_2,j}$. Since $\phi_{i_1}\neq \phi_{i_2}$ and they are both non-zero, $D_{i_1,j}$ and $D_{i_2,j}$ can be solved using these two equations.

\item Parity group $N^*$ loses two symbols, which are the parity symbols $D_{r-1,N^*}$ and $P_{N^*}$. This case is trivial since all data symbols have been directly recovered. 

\item Parity group $N^*$ loses two symbols, which are the parity symbols $P_{N^*}$ and a data symbol $D_{i,N^*}$, $1\leq i\leq r-2$. By eliminating the symbols in the other parity group in $D_{r-1,N^*}$ using (\ref{eqn:longMDS}), and then eliminating $D_{k,N^*}$,  $k\neq i$, we obtain $D_{i,N^*}$. 

\item Parity group $N^*$ loses two symbols, which are the parity symbols $D_{r-1,N^*}$ and a data symbol $D_{i,N^*}$,  $1\leq i\leq r-2$. Note that $P^{N^*}$ is still available and  
\begin{align}
P_{N^*}=\sum_{i=1}^{r-1}D_{i,N^*}=\sum_{i=1}^{r-2}\phi_i\sum_{j=1}^{N^*} D_{i,j}+\phi_{r-1}\sum_{j=1}^{N^*-1}D_{r-1,j}+\sum_{i=1}^{r-1}D_{i,N^*}.
\end{align}
By eliminating the symbols in the other parity group in $P_{N^*}$ and then eliminating $D_{k,N^*}$, $k\neq i$, we obtain $(\phi_{i}+1)D_{i,N^*}$ for some $1\leq i\leq r-2$, and since $\phi_{i}+1\neq 0$ for such $i$, $D_{i,N^*}$ can be correctly obtained. 
\end{enumerate}

The code construction is illustrated in Fig. \ref{fig:code1}.
In the disk repair and data reconstruction procedure given above, we have inherently assumed that the following two facts hold:
\begin{itemize}
\item \textbf{Fact one:} During the repair, each remaining disk contributes exactly one symbol;
\item \textbf{Fact two:} When two disks are not accessible, only one parity group has two  inaccessible symbols, and the other parity groups each have only one symbol or less inaccessible symbol. 
\end{itemize}
These are indeed true by invoking the basic property of Steiner system, more precisely, that any pair of elements in ${I}_n$ appears exactly in one of the blocks in $\mathcal{B}$.

The long code in the construction is an $(M+1,M)$ systematic MDS code whose parity symbol is specified by (\ref{eqn:longMDS}), and the short code is a $(r,r-1)$ systematic MDS code whose parity symbol is specified by (\ref{eqn:paritygroup}). It should be noted that the coefficients in forming the parity symbols are certainly not unique, and we have only given a convenient choice here. In many cases, the performance of codes is better than time-sharing between MSR and MBR points, however, we leave the detailed analysis to the next section to avoid repetition.



\subsection{A Construction Based on BIBDs $S_\lambda(2,r,n)$}

In this subsection, we generalize the construction previously described to the setting of exact-repair regenerating codes for any positive integer $n$, $d=n-1$ and any $k\leq n-1$, based on BIBDs $S_{\lambda}(2,r,n)$. The validity of the construction relies on application of the Schwarz-Zippel lemma, which is used to show that there exists a valid choice of long MDS code when the alphabet is large than a given threshold. 

First fix a BIBD $({I}_n,\mathcal{B})\in S_{\lambda}(2,r,n)$, and again denote $N_\lambda(2,r,n)$ as $N^*$. First define the quantity
\begin{align}
T(A)=\sum_{B\in\mathcal{B}:|B\cap {A}|\geq 2} |B\cap A|-1,
\end{align}
where ${A}\subset {I}_n$ and $|{A}|=n-k$, then further define
\begin{align}
T=\max_{{A}:{A}\subseteq {I}_n,\,|{A}|=n-k} T(A).
\end{align}
The relevance of this quantity will become clear shortly. 
When $n-k=2$, the definition of BIBDs gives $T=\lambda$. The construction given in the previous subsection belongs to this case with $T=\lambda=1$.
In general, the quantity is dependent on the particular block design, and does not appear to have an explicit formula, however, we shall discuss a bound on this quantity in the next section. 

The code we construct has the following parameters
\begin{align}
\alpha = \frac{\lambda(n-1)}{r-1},\quad \beta = \lambda,\quad M=(r-1)N^*-T=\frac{\lambda n(n-1)}{r}-T.
\end{align}
Note that although $\alpha$ is always an integer, $(n-1)$ is not necessarily a multiple of $r-1$ here, unlike in the previous construction. This implies that $\bar{\alpha}$ may not be an integer.

Let the $M$ information symbols in $\mathbb{F}(q)$ be given in a vector $\vec{d}$, and use it to fill the first $M$ entries in a $(r-1)\times N^*$ matrix $D$ following the column-wise order, {\em i.e.}, the first column (top-down), and the second column, etc.; the rest of the $T$-entries of the matrix are left blank. The code requires a matrix $S$ of size $T\times M$, whose entries are also in $\mathbb{F}(q)$. The matrix $S$ is used to generate the parity symbols for the long MDS code, and we shall specify the condition for $S$ shortly. 



\textbf{Encoding:}

The encoding procedure is similar to the procedure given in the previous subsection, with the only difference being that we first compute the multiplication $S\cdot\vec{d}$ and then fill the rest of $D$ matrix using the resultant $T$ parity symbols in a column-wise manner. 

\begin{figure}[tcb]
  \centering
  \includegraphics[width=10cm]{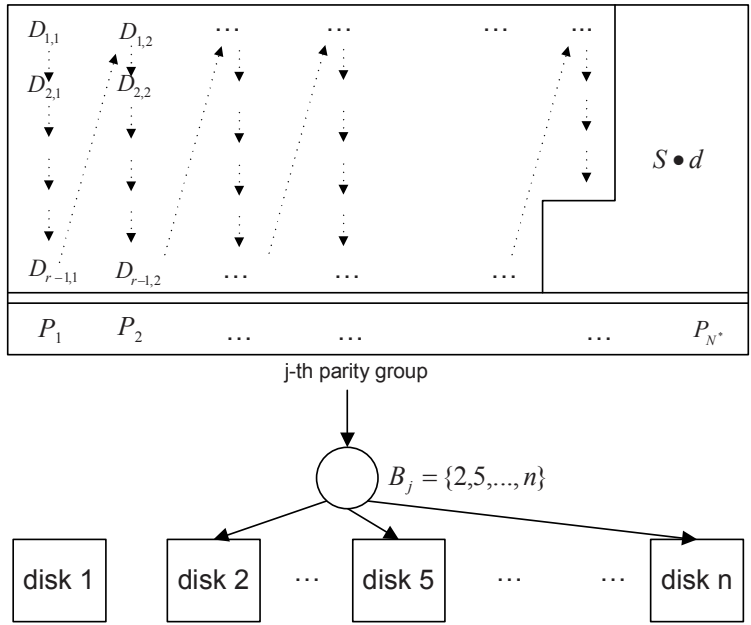}
  \caption{Code structure based on  $S_\lambda(2,r,n)$.\label{fig:code2}}
\end{figure}

\textbf{Repair:}

The repair is precisely the same as the repair procedure given in the previous subsection. Note that each remaining disk contributes exactly $\lambda$ symbols, which is implied by the definition of $S_\lambda(2,r,n)$.

\textbf{Reconstruction:}


Let  $(n-k)$ disks in the set $A$ be inaccessible, where $|A|=n-k$. 
For each parity group $j=1,2,\ldots,N^*$, construct a length-$r$ vector $\vec{z_i}$ as follows
\begin{itemize}
\item If $B_j\cap A \leq 1$: collect, and if necessary, compute using (\ref{eqn:paritygroup}), the symbols $D_{i,j}$, $i=1,2,\ldots,r-1$; let $\vec{z_j}=(D_{1,j},D_{2,j},\ldots,D_{r-1,j},0)^t$;
\item If $B_j \cap A \geq 2$: collect the available symbols in this parity group, denoted as $(D_{i_1,j},D_{i_2,j},\ldots,D_{i_l,j})$, assign $D_{i_1,j},D_{i_2,j},\ldots,D_{i_l,j}$  to the $i_1,i_2,\ldots,i_l$ positions of vector $\vec{z_j}$, and let the rest of $\vec{z_j}$ be zeros.
\end{itemize}
Finally let $\vec{\bar{d}}_A=[\vec{z_1}^t,\vec{z_2}^t,\ldots,\vec{z_{N^*}}^t]^t$, {\em i.e.}, concatenate the vectors $\vec{z_j}$\rq{}s. 
The entries of $\vec{\bar{d}}_A$ are linear combinations of $\vec{d}$. Our claim is that by properly choosing $S$, the vector $\vec{d}$ can be reconstructed from $\vec{\bar{d}}_A$ for any possible set $A$.

\vspace{0.6cm}
This construction is illustrated in Fig. \ref{fig:code2}, from which the difference and similarity from the construction given in the previous subsection is straightforward. For the case $r=2$, the proposed construction is precisely the repair-by-transfer construction in \cite{RashmiShah:12:1}. In this case, the parity symbol $P_j$ is a simple repetition, and the $S_\lambda(2,2,n)$ design is when $\lambda=\lambda^*=1$ in the trivial complete design.

%



Next we show that a matrix $S$ with the desired properties indeed exists. Note that as long as the transfer matrix between $\vec{\bar{d}}_A$ and $\vec{d}$ has rank $M$, the information vector $\vec{d}$ can be correctly reconstructed. To identify this matrix,  first construct a template matrix $R$ of size $r\times (r-1)$ as $R=[I,\vec{1}^t]^t$, where $I$ is the identity matrix, and $\vec{1}$ is the all one vector of length $r-1$. For each $j=1,2,\ldots,N^*$, construct matrix $R_j$ of size $r\times (r-1)$ as follows,
\begin{itemize}
\item If $B_j\cap A \leq 1$, let $R_j$ be $R$ with the last row set to all zeros;
\item If $B_j \cap A =l\geq 2$, then let the corresponding symbols in the $i$-th parity group stored on disks $B_j \setminus A$ be $D_{i_1},D_{i_2},\ldots,D_{i_l}$. Keep the rows $i_1,i_2,\ldots,i_l$ in $R$, and assign the other rows as all zeros, and let the resultant matrix be $R_j$.
\end{itemize}
Finally form a matrix $Q_A$ of size $(rN^*)\times(r-1)N^*$ using matrix $R_j$\rq{}s as the diagonal, {\em i.e.},
\begin{align}
Q_A=\left[
\begin{array}{cccc}
R_1&&&\\
&R_2&&\\
&&\ddots&\\
&&&R_{N^*}
\end{array}
\right]
\end{align}
Clearly, we have
\begin{align}
Q_A\cdot G\cdot \vec{d}=\vec{\bar{d}}_A,
\end{align}
where $G=[I, S^t]^t$. Thus as long as $Q_A\cdot G$ has rank $M$ for each set $A\subset I_n$ such that $|{A}|=n-k$, the information vector $\vec{d}$ can be correctly decoded no matter which $n-k$ disks are inaccessible. We have the following proposition.



\begin{prop}
Among the $q^{TM}$ distinct assignments of $S$, at most a fraction of $q^{-1}{{n\choose k}TM}$ may induce a matrix $Q_A\cdot G$ with rank less than $M$ for some $A\subset I_n$ such that $|{A}|=n-k$.
\end{prop}
\begin{proof}
The proof is a direct application of the Schwartz-Zippel lemma in its counting form. For each $A\subset I_n$ such that $|{A}|=n-k$, if we can show that the fraction of assignments resulting in $\mbox{rank}(Q_A\cdot G)<M$ is bounded by $q^{-1}TM$, then the bound given in the proposition is obtained by a simple union over all choices of $A$. To show this, first remove the all-zero rows in $Q_A$, and then remove the first $T-T(A)$ rows in the remaining matrix, resulting in a matrix  $Q\rq{}_A$. Note that $Q\rq{}_A\cdot G$ is of size $M\times M$, and thus as long as $Q\rq{}_A\cdot G$ has full rank, the matrix $Q_A\cdot G$ has rank $M$. However, $Q\rq{}_A\cdot G$ having full rank is equivalent to $\det(Q\rq{}_A\cdot G)\neq 0$. Since $\det(Q\rq{}_A\cdot G)$ is a polynomial $g(\cdot)$ of the entries of $S$, as long as $g(\cdot)$ is not identically zero, we can apply the Schwartz-Zippel lemma and conclude the proof. The polynomial $g(\cdot)$ is indeed not identically zero, which is proved in the appendix.
\end{proof}

As a consequence of this proposition, when $q>{n \choose k} TM$, there exists at least one valid choice of matrix $S$; in fact, when $q$ is sufficiently large, almost all the assignments of $S$ are valid. The problem of explicitly constructing the matrix $S$ is open, however it may not be as complex as it seems. One possible approach is to  let $S$ be the parity portion of a systematic MDS code generator matrix, and then check whether the full rank conditions are satisfied for each possible set $A\subset I_n$ with $|A|=n-k$, which is a total of ${n \choose k}$ conditions.

\subsection{A Construction Based on $t$-Designs $S_{\lambda}(t,r,n)$}

The code construction for exact-repair regenerating codes presented in the previous section can be generalized to the case $d<n-1$, by using general $t$-design instead of BIBDs. The resulting codes may require different amounts of data contributions from disks during repair, and thus do not strictly belong to the class of codes defined in Section \ref{sec:definition}. For this reason, instead of considering per-disk rate $\beta$ during repair, we shall only consider total repair bandwidth $\gamma$ here. A special class of code, based on complete block designs $S_{\lambda^*}(t,r,n)$, can be made symmetric by time-sharing among different repair rate allocations, as shall be discussed shortly. 

Given a particular $t$-design $(I_n,\mathcal{B})\in S_{\lambda}(t,r,n)$, we shall construct an exact-repair regenerating codes of parameter $(n,k,d)$ using $(X,\mathcal{B})$, where $d=n-t+1$ and $k\leq d$. Similarly as in the last subsection, define the quantity
\begin{align}
T(A)=\sum_{B\in\mathcal{B}:|B\cap {A}|\geq t} |B\cap A|-t+1,
\end{align}
where ${A}\subset {I}_n$ and $|{A}|=n-k$, then further define
\begin{align}
T=\max_{{A}:{A}\subseteq {I}_n,\,|{A}|=n-k} T(A).
\end{align}

The code we construct has the following parameters (note instead of $\beta$, here $\gamma$ is given)
\begin{align}
\label{eqn:performancegeneral}
\alpha = \frac{\lambda{{n-1} \choose {t-1}}}{{{r-1} \choose {t-1}}},\quad \gamma = {(r-t+1)\alpha},\quad M=(r-t+1)N_\lambda(t,r,n)-T.
\end{align}

\begin{figure}[tcb]
  \centering
  \includegraphics[width=10cm]{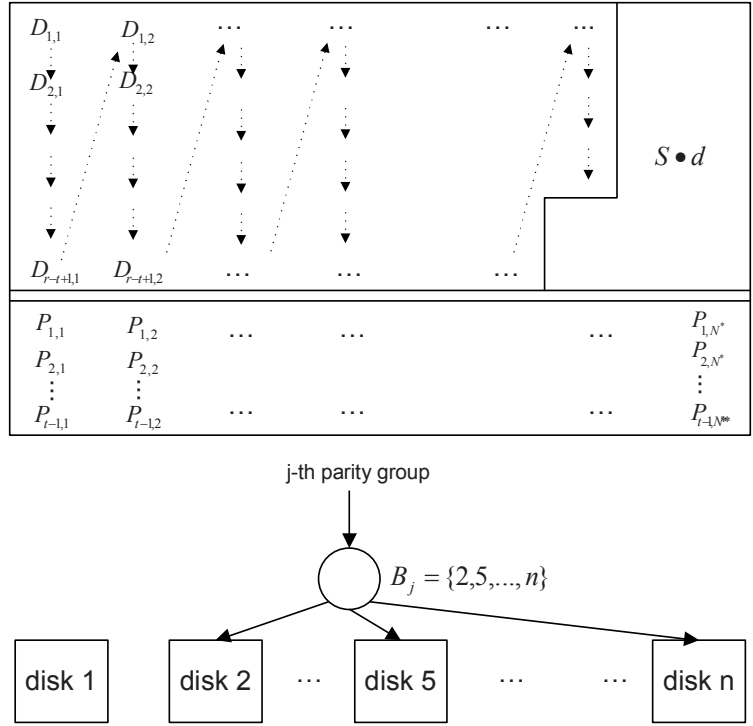}
  \caption{Code structure based on  $S_\lambda(t,r,n)$.\label{fig:code3}}
\end{figure}

The difference from the construction given in the last section is that instead of only one parity symbol, $t-1$ parity symbols $P_{1,j},P_{2,j},\ldots,P_{t-1,j}$ are generated in the parity group $j$, using a fixed systematic $(r,r-t+1)$ MDS code; see Fig. \ref{fig:code3}. Any symbol on a failed disk has a maximum of $n-d=t-1$ symbols from the same parity group that are not participating in the repair, however the $(r,r-t+1)$ MDS code guarantees that this symbol can be recovered using the remaining at least $r-t+1$ symbols on the other disks. Note that the amounts of data contributions from these disks are in general not symmetric, although there may be many choices to choose which symbols to use in the repair. Similarly as in the previous case, it can be shown that there exists a matrix $S$ of size $T\times M$ which guarantees correct decoding in a sufficiently large alphabet, and thus we omit the details to avoid repetition.

One particularly interesting case is when the complete block design $S_{\lambda^*}(t,r,n)$ is used. In this case, although the data contributions from the disks during repair may not be symmetric, one can always time-share among different helper rate contribution allocations, due to the symmetry of the code. Thus this time-sharing version of the codes based on complete block design indeed belongs to the class of codes defined in Section \ref{sec:definition}.

\section{Performance Analysis}
\label{sec:performance}
In this section, we analyze the performance of the proposed codes more systematically. Recall that the quantity $T$ includes an optimization problem, and it is block design dependent. Thus in general $(\alpha,\beta,M)$ of a particular code can not be explicitly evaluated. However, for the codes based on complete block designs, the performance can indeed be explicitly evaluated. Moreover, for any given $(r,n)$, a complete block design $S_{\lambda^*}(2,r,n)$ in fact offers the best performance among all possible BIBDs $S_{\lambda}(2,r,n)$, in terms of the normalized measure $(\bar{\alpha},\bar{M})$, which is shown next.

\subsection{The Optimality of Complete Block Designs}

\begin{prop}
Given an $(n,k,d)$ code $\mathcal{C}_1$ constructed using $t$-design of $(I_n,\mathcal{B})\in S_{\lambda}(n-d+1,r,n)$, which achieves $(\bar{\alpha},\bar{M}_1)$, and an $(n,k,d)$ code $\mathcal{C}_2$ constructed using the complete design $S_{\lambda^*}(n-d+1,r,n)$, which achieves $(\bar{\alpha},\bar{M}_2)$. Then $\bar{M}_1\leq \bar{M}_2$, with inequality holds if and only if in $(I_n,\mathcal{B})$ the quantity $T_{A}$ is uniform for all set $A\subset I_n$ with $|A|=n-k$.
\end{prop}
\begin{proof}

The $T$  function on both $(I_n,\mathcal{B})\in S_{\lambda}(n-d+1,r,n)$ and on the complete design  $S_{\lambda^*}(n-d+1,r,n)$ need to be considered, and in order to distinguish them, we shall the latter as $T^*$.

Consider the block design resulting from a permutation $\pi$ of the elements of $I_n$, which also operates on the blocks in $(I_n,\mathcal{B})$, and denoted the $T(A)$ function on this permuted block design as $T_\pi(A)$, and the $T$ function on this permuted block design as $T_\pi$. Construct a new and larger block design by taking all the blocks resulting through the $n!$ permutation of the design $(I_n,\mathcal{B})$; note that there might be repetition of the same blocks, which is allowed in this new design. We shall denote the $T(A)$ function operating on this new design as $T_p(A)$, and the corresponding $T$ function as $T_p$. 

This new compound design is apparently a complete block design where each block is repeated $\frac{n!\lambda}{\lambda^*}$ times, and thus 
\begin{align}
T_p=\frac{n!\lambda}{\lambda^*}T^*.
\end{align}
However, since it is the combination of $n!$ permutation of the block design  $(I_n,\mathcal{B})$,
we also have
\begin{align}
T_p=T_p(I_{n-k})=\sum_{\pi} T_\pi(I_{n-k})\leq \sum_{\pi} T_\pi=\sum_{\pi} T=n! T.
\end{align}
Thus $\lambda^* T\geq \lambda T^*$, which when combined with (\ref{eqn:N}) and (\ref{eqn:performancegeneral}), gives $\bar{M}_1\leq \bar{M}_2$. Clearly equality holds if and only if $T_\pi(I_{n-k})=T$ for all $\pi$, which is equivalent to $T(A)=T$ for all $A\subset I_n$ with $|A|=n-k$. The proof is complete.
 
\end{proof}

Although complete block designs provide the best $(\bar{\alpha},\bar{M})$ among all the $t$-designs in the same class, other incomplete block designs may lead to simpler code, as illustrated in the following example. 

\textbf{Example:} Consider a code based on the complete block design $S_{7}(2,3,9)$, and thus $T=\lambda^*=7$. Using the general code construction based on BIBDs, we have a $(9,7,8)$ exact-repair regenerating code with
\begin{align}
\alpha = 28,\quad \beta = 7, \quad M=161.
\end{align}
The $S$ matrix for the first layer code in this case is of size $7\times 161$, {\em i.e.}, $7$ parity symbols generated by $161$ information symbols. In contrast, in the example given in Section \ref{sec:example}, also a $(9,7,8)$ exact-repair regenerating code, has a first layer code with only a single parity symbol, generated by $23$ information symbols.  Note however both code achieve the same normalized measure $(\bar{\alpha},\bar{M})=(4,23)$.

\subsection{Performance Analysis Using Complete Block Designs}

Recall that $S_{\lambda^*}(t,r,n)$ can be used to construct codes with different $k$ values, where $t=n-d+1$. With complete block designs, the value $T$ can be explicitly evaluated as follows using the symmetry
\begin{align}
T=T(I_{n-k})&=\sum_{B\in\mathcal{B}:|B\cap I_{n-k}|\geq t} |B\cap I_{n-k}|-t+1\nonumber\\
&=\sum_{B\in\mathcal{B}:|B\cap I_{n-k}|\geq n-d+1} |B\cap I_{n-k}|-n+d\nonumber\\
&=\sum_{i=t}^{\min(n-k,r)}(i-n+d){{n-k}\choose i}{k\choose {r-i}}
\triangleq T_c.
\end{align}

It is clear that the code has a normalized $\bar{\alpha}$ as
\begin{align}
\label{eqn:finalalpha}
\bar{\alpha} = \frac{d}{r-n+d},
\end{align}
and a normalized $\bar{M}$ as
\begin{align}
\bar{M}=\frac{nd}{r}-\frac{dT_c}{ (r+d-n){{n-1}\choose {r-1}}}.
\label{eqn:finalM}
\end{align}
Clearly for $\bar{\alpha}\geq 0$, we need $r\geq n-d+1$. Since at the MSR point, $\bar{\alpha}=d-k+1$, it is more meaningful to choose
\begin{align}
r\leq n-d+\frac{d}{d-k+1}.
\end{align}
However, choosing $r$ greater than this value is also valid, which may yield codes that although not efficient in terms of $(\bar{\alpha}, \bar{M})$, but nevertheless useful due to its simplicity.


\begin{figure}[tcb]
  \centering
  \includegraphics[width=10cm]{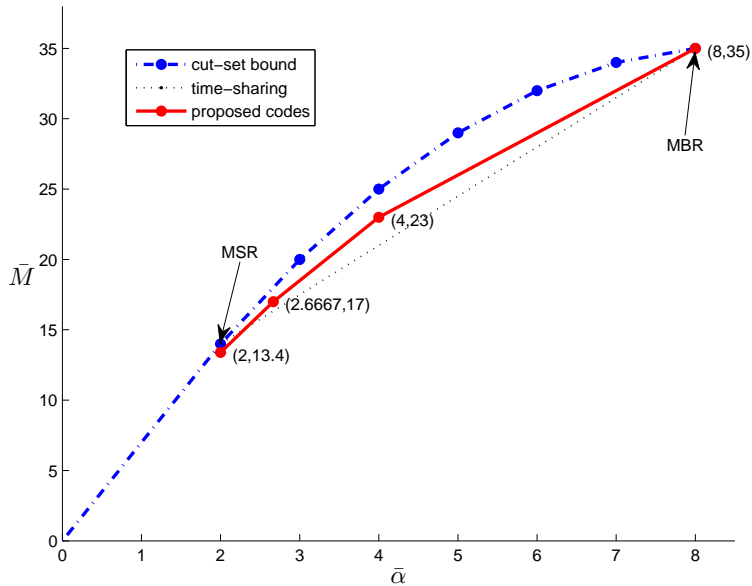}
  \caption{Cut-set bound, time-sharing line and the performance of the proposed codes for  $(n,k,d)=(9,7,8)$.\label{fig:example1}}
\end{figure}

%

\begin{figure}[tcb]
  \centering
  \includegraphics[width=15cm]{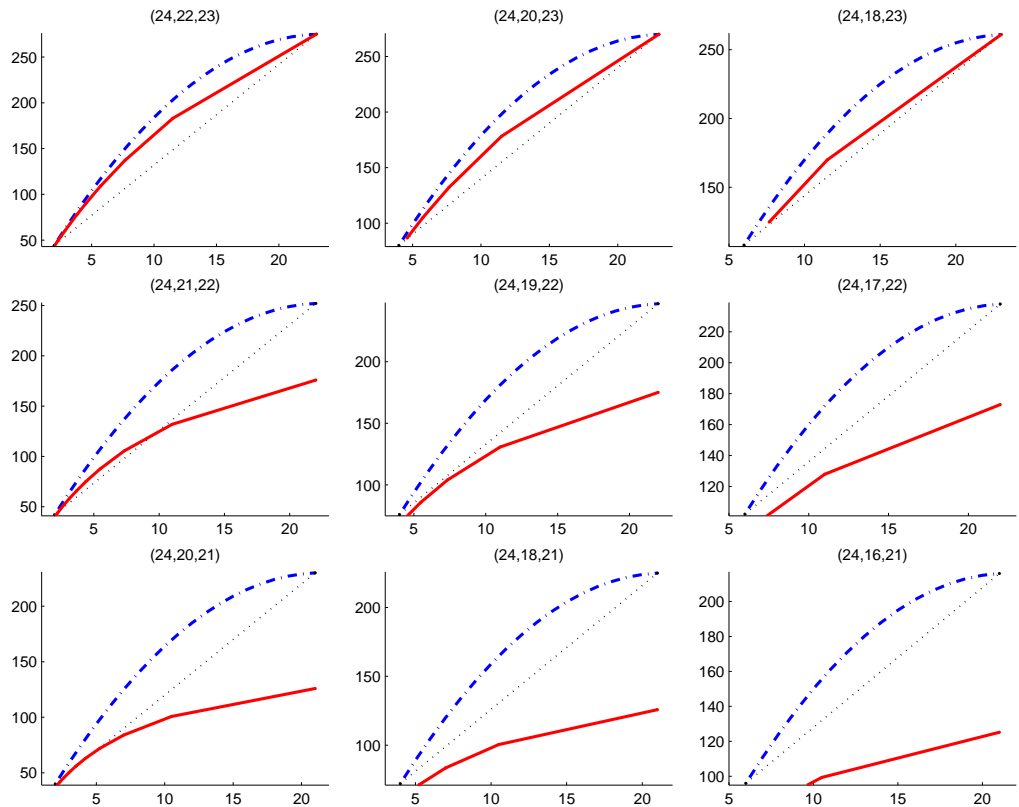}
  \caption{Performance of the proposed codes for different $(k,d)$ parameters when $n=24$.\label{fig:example2} The dashed blue lines are the cut-set bounds, the dotted black lines are the time-sharing lines, and the red solid lines are the tradeoff achieved by the proposed codes. }
\end{figure}

There does not seem to be any simplification for specific $(n,k,d)$ parameters. We provide a few examples to illustrate the performance of the code for various $(n,k,d)$. In Fig. \ref{fig:example1}, we plot the performance of the proposed codes for the case of $(n,k,d)=(9,7,8)$, and for reference the cut-set bound and time-sharing line are also included. It can be seen that in addition to the code example given in Section \ref{sec:example}, there is one more parameter $r=4$ that yields a performance above the time-sharing line; the proposed code also achieves the point $(8,35)$, which is not surprising since in this case it reduces to the optimal construction in \cite{RashmiShah:12:1}. The operating point $(\bar{\alpha},\bar{M})=(2,13.4)$ is worth noting, because although it is not as good as the MSR point $(2,14)$, the penalty is surprisingly small. This suggests that the proposed codes may even be a good albeit not optimal choice to replace an MSR code, particularly when such MSR codes have high complexity.

In Fig. \ref{fig:example2} we plot the performance of codes for different parameters $(k,d)$ when $n=24$. It can be seen that when $d=n-1=23$, the performance is the most competitive, and often superior to the time sharing line. As $d$ value decreases, the method become less effective in terms of its $(\bar{\alpha},\bar{M})$, and becomes worse than the time-sharing line. For the same $d$ value, the code is most effective when $k$ is large, and becomes less so as $k$ value decreases. It should be noted that the lower left corner is the MSR point, and in a wide range of parameters the proposed scheme in fact operates rather close to it, despite the simple coding structure.

\subsection{An Asymptotic Analysis}

\begin{figure}[tcb]
  \centering
  \includegraphics[width=10cm]{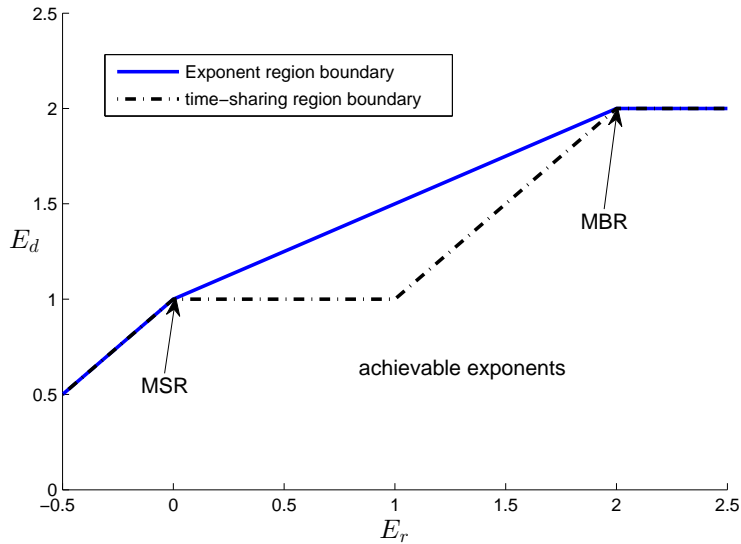}
  \caption{The achievable exponent region $\mathcal{E}$ and the exponent region achieved by time-sharing.\label{fig:asymptotics}}
\end{figure}

In this subsection, we consider the asymptotic performance of the proposed codes when $n$ is large. Recall the case under consideration is when $k=n-\tau_1$ and $d=n-\tau_2$ where $\tau_1$ and $\tau_2$ are fixed constant integers. We have the following theorem.
\begin{theorem}
Let $\mathcal{E}^*$ be the collections of $(\mathsf{E}_r,\mathsf{E}_d)$ pair such that
\begin{align}
\mathsf{E}_d\leq \mathsf{E}_r+1,\quad 2\mathsf{E}_d\leq 2+\mathsf{E}_r,\quad \mathsf{E}_d\leq 2, 
\end{align}
Then $\mathcal{E}^*=\mathcal{E}$, and moreover,  $\mathcal{E}^*$ can be asymptotically achieved by the proposed code construction. 
\end{theorem}
\begin{proof}
We first show that $\mathcal{E}\subseteq\mathcal{E}^*$ by utilizing the cut-set bound in Theorem \ref{theorem:cutset}. For better clarify, we shall write $(\bar{\alpha},\bar{M})$ for a fixed $n$ explicitly as $(\bar{\alpha}^{(n)},\bar{M}^{(n)})$. First notice that the bound implies that for any integer $c \in [0,k]$
\begin{align}
\bar{M}^{(n)}\leq \sum_{i=0}^{k-1}\min(\bar{\alpha}^{(n)},(d-i))\leq \sum_{i=0}^{c-1}\bar{\alpha}^{(n)}+\sum_{c}^{k-1}(d-i)=c\bar\alpha^{(n)}+\frac{(2d-k-c+1)(k-c)}{2}.\label{eqn:boundrewrite}
\end{align}
Taking $c=0$ gives
\begin{align}
\bar{M}^{(n)}\leq \frac{(2d-k+1)k}{2},
\end{align}
which implies that
\begin{align} 
\mathsf{E}^{(n)}_d&\leq \frac{\log \bar{M}^{(n)}}{\log n}\leq \frac{\log(2d-k+1)k-\log 2}{\log n}\nonumber\\
&=\frac{\log(2d-k+1)k-\log 2}{\log n}=\frac{\log(d+\tau_1-\tau_2+1)k-\log 2}{\log n},
\end{align}
and thus for any $(\mathsf{E}_r,\mathsf{E}_d)\in\mathcal{E}$, $\mathsf{E}_d\leq 2$. 
By taking $c=k$, we have
\begin{align}
\bar{M}^{(n)}\leq k\bar\alpha^{(n)},\label{eqn:takectobek}
\end{align}
which implies that
\begin{align}
n\bar\alpha^{(n)}-\bar{M}^{(n)}\geq \frac{n-k}{k}\bar{M}^{(n)}.
\end{align}
It follows that
\begin{align}
\mathsf{E}^{(n)}_r\geq \frac{\log(n\bar\alpha^{(n)}-\bar{M}^{(n)})}{\log n}\geq \frac{\log\frac{n-k}{k}\bar{M}^{(n)}}{\log n}=\frac{\log(n-k)-\log k}{\log n}+\mathsf{E}^{(n)}_d,
\end{align}
and thus $\mathsf{E}_r-\mathsf{E}_d\geq -1$ for any $(\mathsf{E}_r,\mathsf{E}_d)\in\mathcal{E}$. 
Next rewrite (\ref{eqn:boundrewrite}) when as follows
\begin{align}
n\bar{\alpha}^{(n)}-\bar{M}^{(n)}&\geq (n-c)\bar{\alpha}-\frac{(2d-k-c+1)(k-c)}{2}\nonumber\\
&\geq (n-c)\frac{\bar{M}}{k}-\frac{(2d-k-c+1)(k-c)}{2},
\end{align}
where the second inequality is due to (\ref{eqn:takectobek}). Rearrange the right hand side of the above inequality to be a quadratic function in $c$, we have
\begin{align}
n\bar{\alpha}^{(n)}-\bar{M}^{(n)}\geq &-\frac{c^2}{2}+\frac{(2d+1-\frac{2\bar{M}}{k})}{2}c-\frac{(2d-k+1)k}{2}+\frac{n\bar{M}}{k}\nonumber\\
\geq &-\frac{1}{2}(c-\frac{2d+1}{2}+\frac{\bar{M}}{k})^2+\frac{1}{2}(\frac{2d+1}{2}-\frac{\bar{M}}{k})^2-\frac{(2d-k+1)k}{2}+\frac{n\bar{M}}{k}.
\end{align}
When the following condition holds
\begin{align}
\bar{M}^{(n)}\geq (d-k+\frac{1}{2})k,
\end{align}
that is
\begin{align}
\mathsf{E}^{(n)}_d\geq \frac{\log (d-k+\frac{1}{2})k}{\log n}, \label{eqn:Errange}
\end{align}
we can choose 
\begin{align}
c=\left\lfloor\frac{2d+1}{2}-\frac{\bar{M}}{k}\right\rfloor,
\end{align}
and arrive at the bound
\begin{align}
n\bar{\alpha}^{(n)}-\bar{M}^{(n)}\geq &-\frac{1}{2}+\frac{1}{2}(\frac{2d+1}{2}-\frac{\bar{M}}{k})^2-\frac{(2d-k+1)k}{2}+\frac{n\bar{M}}{k}\nonumber\\
&=\frac{\bar{M}^2}{2k^2}+\frac{1}{2}(d-k)^2+\frac{1}{2}(d-k)+(n-d-\frac{1}{2})\frac{\bar{M}}{k}-\frac{3}{8}\nonumber\\
&\geq \frac{\bar{M}^2}{2k^2}-\frac{3}{8}.
\end{align}
This implies that when (\ref{eqn:Errange}) is true, for any sufficiently large $n$ and any $\delta>0$
\begin{align}
\mathsf{E}^{(n)}_r\geq 2\mathsf{E}^{(n)}_d-2-\delta,
\end{align}
and thus $\mathsf{E}_r\geq 2\mathsf{E}_d-2$ when $\mathsf{E}_d\geq 1$. This completes the converse proof. 

For the forward proof, we shall first fix a quantity $0<\epsilon<1$, and consider a sequence of code $n=n_0,n_0+1,\ldots$. From (\ref{eqn:finalalpha}), we get
\begin{align}
\bar{\alpha}^{(n)} = \frac{d}{r-n+d} = \frac{n-\tau_2}{n^\epsilon-\tau_2}. 
\end{align}

Note that 
\begin{eqnarray}
T^* &=&\sum_{i=t}^{\min(n-k,r)}(i-n+d){{n-k}\choose i}{k\choose {r-i}}\\
&\le&\sum_{i=t}^{\min(n-k,r)}(d-k){{n-k}\choose i}{k\choose {r-i}}\\
&=&(\tau_1-\tau_2)\sum_{i=t}^{\min(n-k,r)}{{n-k}\choose i}{k\choose {r-i}}\\
&\le&(\tau_1-\tau_2)\sum_{i=0}^{r}{{n-k}\choose i}{k\choose {r-i}}\\
&=&(\tau_1-\tau_2){{n}\choose r}
\end{eqnarray}
Thus from (\ref{eqn:finalM}) we have
\begin{eqnarray}
\bar{M}^{(n)}&=& \frac{nd}{r}-\frac{dT^*}{ (r+d-n){{n-1}\choose {r-1}}}\\
&\ge& \frac{nd}{r}-(\tau_1-\tau_2)\frac{d{{n}\choose r}}{ (r+d-n){{n-1}\choose {r-1}}}\\
&=& \frac{nd}{r}-(\tau_1-\tau_2)\frac{dn}{ (r+d-n)r}\\
&=& \frac{nd}{r}\left(1-\frac{\tau_1-\tau_2}{ (r-\tau_2)}\right)\\
&=& n^{1-\epsilon} (n-\tau_2) \left(1-\frac{\tau_1-\tau_2}{ (r-\tau_2)}\right)\\
&\geq&n^{1-\epsilon} (n-\tau_2).
\end{eqnarray}
This implies that for any $\delta>0$ and sufficiently large $n$, there exists a code using the proposed design such that $\mathsf{E}^{(n)}_d\geq 2-\epsilon-\delta$. 

Furthermore, we have
\begin{eqnarray}
 n \bar{\alpha}^{(n)} - \bar{M}^{(n)}\le \frac{n^{1-\epsilon}\tau_1(n-\tau_2)}{n^\epsilon - \tau_2},
\end{eqnarray}
which implies that for any $\delta>0$ and sufficiently large $n$, there exists a code using the proposed design such that $\mathsf{E}^{(n)}_r \leq 2-2\epsilon+\delta$. 
Because the region $\mathcal{E}$ is a closed set, it follows that the pair $(\mathsf{E}_r,\mathsf{E}_d)=(2-2\epsilon,2-\epsilon)$ is achievable for any $1>\epsilon>0$, and thus the region $2\mathsf{E}_d\leq 2+\mathsf{E}_r$ is achievable for any $2>\mathsf{E}_r>0$. The case $(\mathsf{E}_r,\mathsf{E}_d)=(0,1)$ can be simply addressed by taking a sequence of $\epsilon_{m}$ such that $\epsilon_m\rightarrow 1$; the case of $(\mathsf{E}_r,\mathsf{E}_d)=(2,2)$ can be addressed similarly. The regions when $\mathsf{E}_r>2$ and $\mathsf{E}_r<0$ are degenerate and they are easily shown achievable either by increasing unnecessarily the redundancy in the code, or increasing unnecessarily the amount of repair bandwidth. The proof is thus complete.  
\end{proof}

For comparison, let us consider the time-sharing scheme between an MSR and an MBR code. Note that the MSR point corresponds to $(\mathsf{E}_r,\mathsf{E}_d) = (0,1)$ and MBR point corresponds to $(\mathsf{E}_r,\mathsf{E}_d)=(2,2)$. For a code using a time-sharing weight $\theta^{(n)}$, the rate can be bounded as
\begin{align}
\bar{M}^{(n)} &=  \theta^{(n)}(n-\tau_1)(\tau_1-\tau_2+1) + (1- \theta^{(n)})\frac{(n-\tau_1)(n+\tau_1-2\tau_2+1)}{2}\nonumber\\
&\leq (\tau_1-\tau_2+1)n+\frac{1- \theta^{(n)}}{2}(n^2-n-\tau_1(\tau_1-2\tau_2+1)),\label{eqn:m1}
\end{align}
and the redundancy can be bounded as
\begin{align}
n\bar{\alpha}^{(n)}-\bar{M}^{(n)} &= \theta^{(n)}\tau_1(\tau_1-\tau_2+1) + (1-\theta^{(n)})(n(n-\tau_2) - \frac{(n-\tau_1)(n+\tau_1-2\tau_2+1)}{2})\nonumber\\
&\geq \frac{1- \theta^{(n)}}{2}(n^2-n-\tau_1(\tau_1-2\tau_2+1)).\label{eqn:m2}
\end{align}
For any sequence of such a time-sharing codes index by $n$, we have 
\begin{align}
\mathsf{E}^{(n)}_r\geq \frac{\log \left(\bar{M}^{(n)}- (\tau_1-\tau_2+1)n\right)}{\log n},
\end{align}
which implies that for any $\delta>0$ and sufficiently large $n$,
\begin{align}
\mathsf{E}^{(n)}_r\geq \mathsf{E}^{(n)}_d-\delta,
\end{align}
and it follows that $\mathsf{E}_r\geq \mathsf{E}_d$ when $\mathsf{E}_d\geq 1$. Using (\ref{eqn:m1}) and (\ref{eqn:m2}), we can also write
\begin{align}
\mathsf{E}^{(n)}_d\leq \frac{\log(n\bar{\alpha}^{(n)}-\bar{M}^{(n)}+(\tau_1-\tau_2+1)n)}{\log n},
\end{align}
which implies that when $\mathsf{E}_r\leq 1$, we must have $\mathsf{E}_d\leq 1$ using the time-sharing strategy.  See Fig. \ref{fig:asymptotics} for an illustration of this region. 

The gap between $\mathcal{E}$ and the time-sharing scheme shows that the improvement of our proposed codes over the time-sharing scheme increases with $n$, and it can be unbounded.

\section{Conclusion}
\label{sec:conclusion}

A new construction for $(n,k,d)$ exact-repair regenerating codes is proposed by combining two layers of error correction codes together with combinatorial block designs. The resultant codes have the desirable \lq\lq{}uncoded repair\rq\rq{} property where the nodes participating in the repair simply send certain stored data without performing any computation. We show that the proposed code is able to achieve performance better than the time-sharing between an MSR code and an MBR code for some parameters. For the case of $d=n-1$ and $k=n-2$, an explicit construction is given in a finite field $\mathbb{F}(q)$ where $q$ is greater or equal to the block size in the combinatorial block designs. For more general  $(d,k)$ parameters, we show that there exist systematic linear codes in a sufficiently large finite field.

\appendix

In this appendix, we prove $\det(Q\rq{}_A\cdot G)$, as a function of the entries of the matrix $S$, is not identically zero. For this purpose, we shall revisit the matrix $Q_A$. Set the first $T-T(A)$ non-zero rows in $Q_A$ to be all zeros, and denote the resulting matrix as $Q^*_A$; let us omit the subscript $A$ from here on for simplicity. If there exists an assignment of $S$ such that $Q^*\cdot G$ has rank $M$, then clearly $\det(Q\rq{}\cdot G)\neq 0$, since $Q\rq{}$ is simply $Q^*$  without the all-zero rows.

Recall the matrix $Q^*$ is of size $(rN^*)\times (M+T)$ with $M$ non-zero rows, and the matrix $G$ is of size $(M+T)\times M$, where $M+T=(r-1)N^*$. 
Let the quotient and the remainder of $M$ divided by $(r-1)$ be $a$ and $b$, respectively, {\em i.e.}, $M=a(r-1)+b$. Let us partition the matrix $Q^*$ into four sub-matrices as
\begin{align}
Q^*=\left[\begin{array}{cc}Q_{11}&Q_{12}\\
Q_{21}&Q_{22}\end{array}\right]=\left[\begin{array}{cc}Q_{11}&0\\
Q_{21}&Q_{22}\end{array}\right]
\end{align}
where $Q_{11}$ is of size $(ar+b)\times(a(r-1)+b)$, which implies $Q_{12}$ is an all-zero matrix. It follows that
\begin{align}
Q^*\cdot G=\left[\begin{array}{cc}Q_{11}&0\\
Q_{21}&Q_{22}\end{array}\right]\cdot \left[\begin{array}{c}I\\S\end{array}\right]=\left[\begin{array}{c}Q_{11}\\Q_{21}+Q_{22}\cdot S\end{array}\right].
\end{align}

In order to show  $Q^*\cdot G$ has rank $M$, our plan is to specify an auxiliary matrix $H$, which satisfies the following two conditions.
\begin{itemize}
\item Condition one: the matrix $[Q^t_{11},H^t]$ has rank $M$;
\item Condition two: the equation $H=Q_{21}+Q_{22}\cdot S$ has a valid solution for $S$.
\end{itemize}
Clearly, if both these two conditions hold, the proof is essentially complete. 

We start by first assuming that there exists at least one all-zero row in the bottom $r-b$ rows of $R_{a+1}$; the other case will be addressed shortly. 
A set of $(a+1)$ intermediate matrics $H_1,H_2,\ldots, H_{a+1}$ shall be constructed as follows. For $j=1,2,\ldots,a$, find the all-zero rows in $R_j$, and denote the indices as $l_1,l_2,\ldots,l_{e_j}$; if $e_j\geq 2$, then the matrix $H_j$ is of size $(e_j-1)\times M$, where the $i$-th row has all zeros except the $(j-1)(r-1)+l_i$ position, which is assigned $1$. For $j=a+1$, find the all-zero rows in the first $b$ rows of $R_{a+1}$, denote the indices as $l_1,l_2,\ldots,l_{e_{a+1}}$; the matrix $H_{a+1}$ is of size $e_{a+1}\times M$, where the $i$-th row has all zeros except the $a(r-1)+l_i$ position, which is assigned $1$. The matrix $H$, which has the same size as   $Q_{21}+Q_{22}\cdot S$, if formed by first assigning all zeros to the rows that are all zeros in $[Q_{21},Q_{22}]$, then assign the rows of $H_1,H_2,\ldots, H_{a+1}$ into the remaining rows of $H$ in any order.

We have inherently assumed above that the total number of rows in $H_1,H_2,\ldots, H_{a+1}$ is the same as the number of rows in $[Q_{21},Q_{22}]$ that have non-zero entries. This is indeed true because the former together with the number of rows in $[Q_{11},Q_{12}]$ that are have non-zero entries totals to $M$, while the latter also satisfies this relation. 

To see that condition one holds, notice that in matrix $H$, by exchanging the rows $(j-1)r+l_1,(j-1)r+l_2,\ldots,(j-1)r+l_{e_j-1}$ in $Q_{11}$ ({\em i.e.}, the rows corresponding to the first $e_j-1$ all-zero rows in $R_j$) and the rows of $H_i$, each block matrix $R_j$ can have rank $r-1$; by similar operation, the top $b$ rows of matrix $R_{a+1}$ is an identity matrix. Due to the block diagonal structure of the matrix $Q_{11}$, this indeed implies that the matrix $[Q^t_{11},H^t]$ has rank $a(r-1)+b=M$.

To see that condition two also holds, we solve for $S$ block by block. First consider the block $R_{N^*}$ in $Q_{22}$, and assume $N^*>a+1$. Due to the block structure of $Q_{22}$, the determination of the last $(r-1)$ rows of $S$ only depends on $R_{N^*}$ and the last $r$ rows of $H$, but not any other entries in $Q^*$ and $H$. Let us denote the sub-matrix consisting of the last $r-1$ rows of $S$ as $S\rq{}$, and denote the sub-matrix consisting of the last $r$ rows of $H$ as $H\rq{}$. The problem essentially reduces to finding a solution for $R_{N^*}\cdot S\rq{}=H\rq{}$. Denote the indices of the rows which have non-zero entries in $R_{N^*}$ as $i_1,i_2,\ldots,i_e$, then the column span of $R_{N^*}$ is the space spanned by columns with a single $1$ at the $i_1,i_2,\ldots,i_e$ positions; this further relies on the structure of $R_{N^*}$ and the fact that there is at least one all-zero row in it. It follows that the column span of $H\rq{}$ is in the column span of $R_{N^*}$,  and thus there indeed exists a solution for $S\rq{}$. 
Repeat this process for the other blocks, as well as the partial block of $R_{a+1}$ in $Q_{22}$, a solution for $S$ is found. 



The case that there exists no all-zero row in the bottom $r-b$ rows of $R_{a+1}$ introduce the complication that for the partial block of $R_{a+1}$ in $Q_{22}$, because in this case its column span is one-dimension less than the space spanned by columns with a single $1$ at the desired positions. However, the only change required is the following: the $(r-b)$-th row of $H$ is chosen to be the summation of its first $(r-b-1)$ rows and the $(r-b)$-th row of $Q_{21}$. It is straightforward to check that both condition one and condition two can still be made to hold for this case. The proof is complete. \QED

\bibliographystyle{IEEEbib}

\end{document}